\documentclass[11pt,a4paper]{article}
\usepackage[utf8]{inputenc}

\usepackage[margin=1in]{geometry}

\usepackage{amsmath}
\usepackage{amsfonts}
\usepackage{amssymb}
\usepackage{amsthm}
\usepackage{graphicx}

\newtheorem{lem}{Lemma}
\newtheorem{prop}{Proposition}

\newtheorem{thm}{Theorem}

\title{A Duality Result for Robust Optimization\\with Expectation Constraints}

\author{
 Christopher W. Miller\thanks{Department of Mathematics, University of California, Berkeley ({miller@math.berkeley.edu}). Supported in part by NSF GRFP under grant number DGE 1106400.}
}

\date{}

\begin{document}

\maketitle

\begin{abstract}
This paper demonstrates a practical method for computing the solution of an expectation-constrained robust maximization problem with immediate applications to model-free no-arbitrage bounds and super-replication values for many financial derivatives. While the previous literature has connected super-replication values to a convex minimization problem whose objective function is related to a sequence of iterated concave envelopes, we show how this whole process can be encoded in a single convex minimization problem. The natural finite-dimensional approximation of this minimization problem results in an easily-implementable sparse linear program. We highlight this technique by obtaining no-arbitrage bounds on the prices of forward-starting options, continuously-monitored variance swaps, and discretely-monitored gamma swaps, each subject to observed bid-ask spreads of finitely-many vanilla options.
\end{abstract}

%\subjclass{Primary xxxxx; Secondary xxxxx, xxxxx}

\begin{keywords}
Robust optimization, expectation constraints, model-free bounds, super-replication, strong duality.
\end{keywords}

% 60G40 stopping times, optimal stopping problems
% 93E20 optimal stochastic control
% 91G80 financial applications of other theories
% 90C05 linear programming

\begin{AMS}93E20, 91G80, 90C05.\end{AMS}

\pagestyle{myheadings}
\thispagestyle{plain}
\markboth{C. Miller}{A Duality Result for Robust Optimization with Expectation Constraints}

\section{Introduction}

This short paper demonstrates a practical method for concretely computing the solution of an expectation-constrained robust maximization problem by converting to a infinite-dimensional linear program which admits a natural finite-dimensional approximation. The motivation for this problem is largely financial, as it represents a model-free no-arbitrage upper bound on the value of a given derivative subject to known price bounds on finitely-many other derivatives.

The particular problem we consider is the following expectation-constrained robust optimization problem:
\begin{equation}\label{Equation:MainProblem}
\begin{array}{rcl}
\overline{p} := & \sup\limits_{\mathbb{Q}\in\mathcal{Q}} & \mathbb{E}^\mathbb{Q}\left[\sum\limits_{k=1}^n f_k(X_{T_k},\cdots,X_{T_{k-d}})\right] \\
& \text{s.t.} & \mathbb{E}^\mathbb{Q}\left[g_k(X_{T_k},\ldots,X_{T_{k-d}})\right] \geq 0\text{ for each }k\in\{1,\ldots,n\}.
\end{array}
\end{equation}
We take $f_k:\mathbb{R}^{d+1}\to\mathbb{R}$ and $g_k:\mathbb{R}^{d+1}\to\mathbb{R}^{p_k}$ to be given continuous functions for each $k\in\{1,\ldots,n\}$. We take $T_{1-d}<\cdots<T_0=0<T_1<\cdots<T_n$ as a fixed time-discretization. We fix a convex set $C\subset\mathbb{R}$ and values $x_0,\ldots,x_{1-d}\in C$. Then we let $\mathcal{Q}$ represent the collection of all probability measures $\mathbb{Q}$ under which $\{X_{T_k}\}_{k\in\{1-d,\ldots,n\}}$ is a $C$-valued martingale satisfying $X_{T_0}=x_0,\ldots,X_{T_{1-d}}=x_{1-d}$ almost-surely. We assume, for simplicity, that $f_k$ and each component of $g_k$ can be bounded from above by (possibly different) affine functions for each $k\in\{1,\ldots,n\}$.\footnote{This assumption rules out the case $\overline{p}=+\infty$ and simplifies analysis. This can be relaxed by examining where it shows up in the proof, but at the cost of significantly more work. This assumption is generally satisfied in practice, either from capping the pay-off of some financial derivative or assuming bounds on the values of $X$.} For later convenience, we denote $p:=\sum_{k=1}^n p_k$.

Our financial motivation comes from interpreting the functions $f_1,\ldots,f_n$ as the pay-off of some derivative to be super-replicated, while the functions $g_1,\ldots,g_n$ encode the pay-offs and known bid-ask spreads of a collection of other derivatives which may be used for hedging. This will be made more explicit with several concrete examples in Section~\ref{Section:Examples}.

The heart of this paper is an investigation of a duality relationship between the expectation-constrained robust maximization problem \eqref{Equation:MainProblem} and the following minimization problem:
\begin{equation}\label{Equation:EquivalentProblem}
\begin{array}{rcl}
\overline{d} := & \inf\limits_{\left(\lambda,\phi,h\right)\in\mathcal{A}} & \phi_1(x_0,x_0,\ldots,x_{1-d}) \\
& \text{s.t.} & h_k = f_k + \lambda_k\cdot g_k \text{ for each }k\in\{1,\ldots,n\}\\
&& \phi_n \geq h_n \\
&& \phi_k(y_k,\ldots,y_{k-d}) \geq h_k(y_k,\ldots,y_{k-d}) + \phi_{k+1}(y_k,y_k,\ldots,x_{y-d+1})\\
&& \hspace{1cm}\text{for each }k\in\{1,\ldots,n-1\}\text{ and all }(y_k,\ldots,y_{k-d})\in C^{d+1}\\
&& \phi_k\text{ is concave in its first entry for each }k\in\{1,\ldots,n\}\\
&& \lambda_k \geq 0\text{ component-wise in }\mathbb{R}^{p_k}\text{ for each }k\in\{1,\ldots,n\},
\end{array}
\end{equation}
where
\[\mathcal{A} := \mathbb{R}^{p_1}\times\cdots\times\mathbb{R}^{p_n}\times C_0(C^{d+1},\mathbb{R})^{2n}.\]
While \eqref{Equation:MainProblem} involves maximization over a collection of martingale measures subject to an expectation constraint, \eqref{Equation:EquivalentProblem} involves minimization over continuous functions subject to concavity constraints. This essentially encodes the computation of iterated concave envelopes, a point which is made clear in Section~\ref{Section:Proof}.

The main result of this paper is the following theorem:
\begin{thm}\label{Theorem:MainResult}
$\overline{p}\leq\overline{d}$. Furthermore, if there exists $\mathbb{Q}\in\mathcal{Q}$ such that $\mathbb{E}^\mathbb{Q}\left[g_k\left(X_{T_k},\ldots,X_{T_{k-d}}\right)\right]\geq 0$ component-wise for each $k\in\{1,\ldots,n\}$, then $\overline{p}=\overline{d}$.
\end{thm}

We can interpret Theorem~\ref{Theorem:MainResult} as describing a duality relationship between the two problems, along with a sufficient condition for strong duality. If the strong duality relationship does not hold, then $\overline{d}$ can be interpreted as a super-replication price.

The idea of relating model-free or robust no-arbitrage price bounds to an infinite-dimensional linear programming problem is not new. Many authors have investigated duality relationships between no-arbitrage price bounds and semi-static super-hedging portfolios (see \cite{Neufeld2013,Acciaio2013,Nutz2014,Riedel2011,Deng2016}). However, we emphasize that the approach of this paper applies to many common derivatives, is easily implementable via a finite-dimensional approximation, and directly returns super-hedging portfolios and worst-case price dynamics.

The main technique of this paper is to re-write the maximization over martingale measures instead as a minimization problem. Intuitively, we view this as analogous to viewing a viscosity solution as the minimum viscosity super-solution, which in this case corresponds to computing concave envelopes. This is essentially an alternate perspective on ideas demonstrated in the recent paper by Kahal{\'e} \cite{Kahale2012}, in which the author casts the super-replication of several common exotic derivatives as a convex optimization problem whose objective function involves the computation of iterated concave envelopes. Rather than viewing the problem as a multi-stage optimization which requires specialized numerical routines, we demonstrate how to encode the same methodology in a single minimization problem. The natural finite-dimensional approximation of \eqref{Equation:EquivalentProblem} can be immediately solved by common software packages for linear programs (e.g. Mosek, Matlab, GLPK, et cetera).

%It is important to emphasize the practical importance of this equivalent problem. While \eqref{Equation:MainProblem} is, in principle, a convex optimization problem, the author is not aware of an easily applicable and efficient finite-dimensional approximation besides enumerating all possible joint probability measures. On the other hand, \eqref{Equation:EquivalentProblem} admits an immediate finite-dimensional approximation in the form of a linear program.

In particular, if we let $\Lambda := \{y_1,\ldots,y_m\}$ be a choice of mesh for $C$ by $m$ grid-points whose convex hull contains $\{x_0,\ldots,x_{1-d}\}$, then we obtain a natural finite-dimensional approximation of \eqref{Equation:EquivalentProblem} as
\begin{equation}\label{Equation:ApproximateProblem}
\begin{array}{rcl}
\overline{d}_\Lambda := & \inf\limits_{\left(\lambda,\phi,h\right)\in\mathcal{A}_\Lambda} & \sum\limits_{i_0,\ldots,i_d=1}^m\gamma_{i_0\cdots i_d}\cdot \phi^1_{i_0\cdots i_d} \\
& \text{s.t.} & h^k_{i_0\cdots i_d} = f_k\left(y_{i_0},\ldots,y_{i_d}\right)+\lambda_k\cdot g_k\left(y_{i_0},\ldots,y_{i_d}\right)\text{ for each }k\in\{1,\ldots,n\} \\
&& \hspace{1cm}\text{ and all }1\leq i_0,\ldots,i_d\leq m\\
&& \phi^n_{i_0\cdots i_d} \geq h^n_{i_0\cdots i_d}\text{ for all }1\leq i_0,\ldots,i_d\leq m\\
&& \phi^k_{i_0\cdots i_d} \geq h^k_{i_0\cdots i_d} + \phi^{k+1}_{i_0 i_0\cdots i_{d-1}} \text{ for each }k\in\{1,\ldots,n-1\}\\
&&\hspace{1cm}\text{and all }1\leq i_0,\ldots,i_d\leq m\\
&& (y_{i_0-1}-y_{i_0})\phi^k_{(i_0+1) i_1\cdots i_d}+(y_{i_0+1}-y_{i_0-1})\phi^k_{i_0 i_1\cdots i_d}+(y_{i_0}-y_{i_0+1})\phi^k_{(i_0-1) i_1\cdots i_d} \geq 0\\
&&\hspace{1cm}\text{for all }k\in\{1,\ldots,n\}\text{, all }2\leq i_0\leq m-1\text{, and all }1\leq i_1,\ldots,i_d\leq m\\
&& \lambda_k \geq 0\text{ component-wise in }\mathbb{R}^{p_k}\text{ for each }k\in\{1,\ldots,n\},
\end{array}
\end{equation}
where $\gamma\in\mathbb{R}^{m\times\cdots\times m}$ is taken as an linear interpolation operator of $(x_0,x_0,\ldots,x_{i-1})$ corresponding to the choice of mesh $\Lambda$ and
\[\mathcal{A}_\Lambda := \mathbb{R}^{p_1}\times\cdots\times\mathbb{R}^{p_1}\times \left(\mathbb{R}^{m\times\cdots\times m}\right)^{2n} \simeq \mathbb{R}^{p+2n\times m^{d+1}}.\]

Despite being notationally complicated, the minimization problem \eqref{Equation:ApproximateProblem} is an easily-implementable linear program with $O\left(p + n\times m^{d+1}\right)$ unknowns and inequality constraints. Then, for any fixed $d$, this can be solved in polynomial time\footnote{We do not necessarily claim to obtain an algorithm which outperforms that provided in \cite{Kahale2012}. The main selling-point of our approach is the alternative conceptualization and ease of implementation using standard packages which run ``fast enough'' in practice.} with respect to $n$, $m$, and $p$ using standard algorithms \cite{Khachiyan1980,Karmarkar1984,Nesterov1994}.  Furthermore, it is a sparse linear program with $O\left(p\times n\times m^{d+1}\right)$ non-zero elements, so there are specialized algorithms with even faster performance \cite{Yen2015,Andersen2000}. We leave a complete analysis of convergence of $\overline{d}_\Lambda$ to $\overline{d}$ to interested researchers and instead choose to focus on Theorem~\ref{Theorem:MainResult} and practical applications.

\section{Some Concrete Examples}\label{Section:Examples}

In this section, we provide three concrete examples which illustrate how to apply the results of this paper to obtain model-free no-arbitrage upper bounds for exotic derivatives. The results of this paper are not general enough to encompass many path-dependent derivatives, but generally apply to those whose pay-off depends upon the current value of the underlying along with finitely-many previous values of the underlying.

Although the ideas in the follows sections can easily be applied to similar derivatives, we choose to include computations of upper bounds for forward-starting at-the-money call options, continuously-monitored variance swaps, and discretely-monitored gamma swaps.

\subsection{Forward-Starting Call Option}\label{Subsection:ForwardStartingCall}

We begin with a simple first example. We consider model-free no-arbitrage bounds on the price of a forward-starting at-the-money call option given bid-ask spreads for finitely-many call options at expiring on the two terminal dates. No-arbitrage price bounds for forward-starting options have been obtained in many different settings previously in the literature, such as \cite{Hobson2015,Kahale2012}.

For simplicity, we take the risk-free rate as zero in this example, so $X$ represents the underlying price process. The case of a non-zero deterministic risk-free rate can be covered immediately by re-interpreting $X$ as the discounted-price process of the underlying and modifying all pay-offs accordingly.

We write the model-free no-arbitrage upper bound in the following form:
\begin{equation}\nonumber
\begin{array}{cl}
\sup\limits_{\mathbb{Q}\in\mathcal{Q}} & \mathbb{E}^\mathbb{Q}\left[\left(X_{T_2}-X_{T_1}\right)^+\right] \\
\text{s.t.} & \text{Bid}_{1,\ell} \leq\mathbb{E}^\mathbb{Q}\left[\left(X_{T_1}-\text{Strike}_{1,\ell}\right)^+\right] \leq \text{Ask}_{1,\ell}\text{ for each }\ell\in\{1,\ldots,p_1\}\\
& \text{Bid}_{2,\ell} \leq\mathbb{E}^\mathbb{Q}\left[\left(X_{T_2}-\text{Strike}_{2,\ell}\right)^+\right] \leq \text{Ask}_{2,\ell}\text{ for each }\ell\in\{1,\ldots,p_2\},
\end{array}
\end{equation}
where $\text{Bid}_{k,\ell}$ and $\text{Ask}_{k,\ell}$ represent the market bid-ask spread of a call option with expiration $T_k$ and strike $\text{Strike}_{k,\ell}$ for each $k\in\{1,2\}$ and $\ell\in\{1,\ldots,p_k\}$. We take $C := [0,\infty)$ in this computation.

This immediately translates into the framework of \eqref{Equation:MainProblem} by taking $d=1$, $n=2$, $f_1(y_1,y_0) := 0$, $f_2(y_2,y_1) := (y_2-y_1)^+$, and
\[g_k(y_k,y_{k-1}) := \left(\begin{array}{c}
\left(y_k-\text{Strike}_{k,1}\right)^+ - \text{Bid}_{k,1} \\
\text{Ask}_{k,1} - \left(y_k-\text{Strike}_{k,1}\right)^+ \\
\vdots \\
\left(y_k-\text{Strike}_{k,p_k}\right)^+ - \text{Bid}_{k,p_k} \\
\text{Ask}_{k,p_k} - \left(y_k-\text{Strike}_{k,p_k}\right)^+
\end{array}\right)\]
for each $k\in\{1,2\}$. Of course, we can also obtain a lower bound by taking $\tilde{f}_k := -f_k$.

In the following, we consider the results of a numerical implementation of the corresponding finite-dimensional approximation given by \eqref{Equation:ApproximateProblem}. Here, we take $x_0 = \$100$, $T_1 = 1/6$, and $T_2 = 5/12$. We take $\text{Strike}_k = \{\$70,\ldots,\$130\}$ for each $k\in\{1,2\}$ and generate bid-ask spreads from the Black-Scholes pricing formula with $\sigma = 20\%$. Lastly, in these results, we take
\[\Lambda := \{\$0,\$10,\ldots,\$60,\$70,\$71,\ldots,\$129,\$130,\$140,\ldots,\$190,\$200,\$10000\}.\]
In Table~\ref{Table:ForwardStartingCallSuperHedge}, we illustrate the resulting static hedge positions corresponding to the super-replication strategy for an upper bound. As expected, the super-replicating strategy is long-volatility at $T_2$ and short volatility at $T_1$. The resulting no-arbitrage upper bound on the price of this forward-starting call option is $\$5.2708$, which is corroborated by the results in \cite{Kahale2012}. Similarly, in Table~\ref{Table:ForwardStartingCallSubHedge}, we illustrate the resulting static hedge positions corresponding to the sub-replication strategy for a lower bound. The resulting no-arbitrage lower bound on the price of this forward-starting call option is $\$1.9266$.

\begin{table}[h]
	\centering
	\begin{tabular}{ |c||c|c|c|c|c|c|c| }
		\hline
		\text{Strike} & \$70 & \$80 & \$90 & \$100 & \$110 & \$120 & \$130 \\ \hline\hline
		$\lambda_1$	& 0.0395 & -0.2400 & -0.5000 & -0.4000 & -0.5000 & 0.3179 & -0.0024 \\ \hline
		$\lambda_2$ & 0.4514 & 0.4800 & 0.4200 & 0.4800 & 0.4200 & 0.4800 & 0.2402 \\ \hline
	\end{tabular}
	\caption{Static positions in call options expiring at times $T_1=1/6$ and $T_2=5/12$, corresponding to a super-hedge of an at-the-money forward-starting call option. The corresponding super-replication value is $\$5.2708$.}
	\label{Table:ForwardStartingCallSuperHedge}
\end{table}

\begin{table}[h]
	\centering
	\begin{tabular}{ |c||c|c|c|c|c|c|c| }
		\hline
		\text{Strike} & \$70 & \$80 & \$90 & \$100 & \$110 & \$120 & \$130 \\ \hline\hline
		$\lambda_1$	& 1.2488 & 0.0010 & -0.0020 & -0.7990 & -0.4000 & 0.2000 & 0.0000 \\ \hline
		$\lambda_2$ & -1.2488 & -0.0010 & 0.0020 & 0.7990 & 0.4000 & -0.2000 & 0.0000 \\ \hline
	\end{tabular}
	\caption{Static positions in call options expiring at times $T_1=1/6$ and $T_2=5/12$, corresponding to a sub-hedge of an at-the-money forward-starting call option. The corresponding sub-replication value is $\$1.9266$.}
	\label{Table:ForwardStartingCallSubHedge}
\end{table}

\subsection{Continuously-Monitored Variance Swap}

Next we consider the the problem of obtaining model-free no-arbitrage bounds on the price of a continuously-monitored variance swap. It is well-known that if the underlying price process is a continuous semi-martingale and the risk-free rate is taken to be zero, then we can write
\[\langle \log X\rangle_T = 2\log x_0 + \int_0^T 2 X_t^{-1}\,dX_t - 2\log X_t\]
under any risk-neutral probability measure (see \cite{Carr2010}).

Taking $x_0=\$100$, without loss of generality, we can write the model-free no-arbitrage upper bound on the price of a continuously-monitored variance swap subject to the bid-ask spreads of finitely-many co-terminal call options in the following form:
\begin{equation}\nonumber
\begin{array}{cl}
\sup\limits_{\mathbb{Q}\in\mathcal{Q}} & \mathbb{E}^\mathbb{Q}\left[-2\log\left(X_{T_n}/100\right)\right] \\
\text{s.t.} & \text{Bid}_{k,\ell} \leq\mathbb{E}^\mathbb{Q}\left[\left(X_{T_k}-\text{Strike}_{k,\ell}\right)^+\right] \leq \text{Ask}_{k,\ell}\text{ for each }\ell\in\{1,\ldots,p_k\}\\
&\hspace{1cm}\text{and each }k\in\{1,\ldots,n\},
\end{array}
\end{equation}
where $\text{Bid}_{k,\ell}$ and $\text{Ask}_{k,\ell}$ represent the bid-ask spread of a call option with expiration $T_k$ and strike $\text{Strike}_{k,\ell}$. We take $C := (0,\infty)$ in this computation.

As before, this translates into the framework of \eqref{Equation:MainProblem} by taking $d=0$, $f_k(y_k) =0$ for each $k\in\{1,\ldots,n-1\}$, $f_n(y_n) = -2\log(y_n)$, and
\[g_k(y_k) := \left(\begin{array}{c}
\left(y_k-\text{Strike}_{k,1}\right)^+ - \text{Bid}_{k,1} \\
\text{Ask}_{k,1} - \left(y_k-\text{Strike}_{k,1}\right)^+ \\
\vdots \\
\left(y_k-\text{Strike}_{k,p_k}\right)^+ - \text{Bid}_{k,p_k} \\
\text{Ask}_{k,p_k} - \left(y_k-\text{Strike}_{k,p_k}\right)^+
\end{array}\right)\]
for each $k\in\{1,\ldots,n\}$. Of course, we can also obtain a lower bound by taking $\tilde{f}_n := -f_n$.

In the following, we consider the results of a numerical implementation of the corresponding finite-dimensional approximation given by \eqref{Equation:ApproximateProblem}. Here, we take $x_0 = \$100$, $n=2$, $T_1 = 1/6$, and $T_2 = 5/12$. We take $\text{Strike}_k = \{\$70,\ldots,\$130\}$ for each $k\in\{1,2\}$ and generate bid-ask spreads from the Black-Scholes pricing formula with $\sigma = 20\%$. Lastly, in these results, we take
\[\Lambda := \{\$1,\$10,\$20,\ldots,\$60,\$70,\$71,\ldots,\$129,\$130,\$140,\ldots,\$190,\$200,\$10000\}.\]
In Table~\ref{Table:VarianceSwapSuperHedge}, we illustrate the resulting static hedge positions corresponding to the super-replication strategy for an upper bound. As expected, the super-replicating strategy is long-volatility at $T_2$ and neutral volatility at $T_1$. Notice, for strikes $\$80$ through $\$120$, the super-hedge positions are approximately proportional to $1/K^2$, which matches the theoretical hedge position when all strikes available.

The resulting no-arbitrage upper bound on the price of this variance swap is $\$0.0208$, which may also be expressed in normalized volatility form as $\sqrt{T_2^{-1}\times 0.0208}\approx 22.3\%$. Similarly, in Table~\ref{Table:VarianceSwapSubHedge}, we illustrate the resulting static hedge positions corresponding to the sub-replication strategy for a lower bound. The resulting no-arbitrage lower bound on the price of this variance swap s $\$0.0156$, which can alternatively be expressed as $\sqrt{T_2^{-1}\times 0.0156}\approx 19.3\%$.

\begin{table}[h]
	\centering
	\begin{tabular}{ |c||c|c|c|c|c|c|c| }
		\hline
		\text{Strike} & \$70 & \$80 & \$90 & \$100 & \$110 & \$120 & \$130 \\ \hline\hline
		$\lambda_1$	& 0.0000 & 0.0000 & 0.0000 & 0.0000 & 0.0000 & 0.0000 & 0.0000 \\ \hline
		$\lambda_2$ & 0.0964 & 0.0031 & 0.0025 & 0.0020 & 0.0017 & 0.0014 & 0.0151 \\ \hline
	\end{tabular}
	\caption{Static positions in call options expiring at times $T_1=1/6$ and $T_2=5/12$, corresponding to a super-hedge of a continuously-monitored variance swap. The corresponding super-replication value is $\$0.0208$.}
	\label{Table:VarianceSwapSuperHedge}
\end{table}

\begin{table}[h]
	\centering
	\begin{tabular}{ |c||c|c|c|c|c|c|c| }
		\hline
		\text{Strike} & \$70 & \$80 & \$90 & \$100 & \$110 & \$120 & \$130 \\ \hline\hline
		$\lambda_1$	& 0.0292 & 0.0000 & 0.0000 & 0.0000 & 0.0000 & 0.0000 & 0.0000 \\ \hline
		$\lambda_2$ & -0.0251 & 0.0032 & 0.0024 & 0.0021 & 0.0016 & 0.0014 & 0.0013 \\ \hline
	\end{tabular}
	\caption{Static positions in call options expiring at times $T_1=1/6$ and $T_2=5/12$, corresponding to a sub-hedge of a continuously-monitored variance swap. The corresponding sub-replication value is $\$0.0156$.}
	\label{Table:VarianceSwapSubHedge}
\end{table}

\subsection{Discretely-Monitored Gamma Swap}

Finally, we consider an instance with $n>>1$. Here, we consider the problem of obtaining model-free no-arbitrage bounds on the price of a type of discretely-monitored gamma swap. The pay-off of a gamma swap is typically defined as
\begin{equation}\label{Equation:GammaSwapPayoff}
\sum\limits_{k=1}^n X_{T_k}\log\left(X_{T_k}/X_{T_{k-1}}\right)^2,
\end{equation}
where the extra term $X_{T_k}$ has been added to the pay-off of a discretely-monitored variance swap. This serves multiple purposes. For our purposes, it mainly serves to protect from crash risk in the pay-off without artificially putting a cap on the value of $X$. In practice, this modification is also useful for dispersion trading and expressing views on the volatility skew. For more on gamma swaps, see \cite{Lee2010}.

For the purposes of this paper, we need to slightly modify the definition above further. The pay-off \eqref{Equation:GammaSwapPayoff} is not bounded above by an affine function, so it is impossible to super-hedge with only call and put options. We modify the pay-off slightly to satisfy the desired property\footnote{We leave it to the interested reader to check that $(x\wedge y)\log(x/y)^2 \leq 4e^{-2}(x+y)$ for all $x,y\in(0,\infty)$.}, changing it to
\begin{equation}\nonumber
\sum\limits_{k=1}^n (X_{T_k}\wedge X_{T_{k-1}})\log\left(X_{T_k}/X_{T_{k-1}}\right)^2.
\end{equation}

We then write the model-free no-arbitrage upper bound in the following form:
\begin{equation}\nonumber
\begin{array}{cl}
\sup\limits_{\mathbb{Q}\in\mathcal{Q}} & \mathbb{E}^\mathbb{Q}\left[\sum\limits_{k=1}^n (X_{T_k}\wedge X_{T_{k-1}})\log\left(X_k/X_{k-1}\right)^2\right] \\
\text{s.t.} & \text{Bid}_{k,\ell} \leq\mathbb{E}^\mathbb{Q}\left[\left(X_{T_k}-\text{Strike}_{k,\ell}\right)^+\right] \leq \text{Ask}_{k,\ell}\text{ for each }\ell\in\{1,\ldots,p_k\}\\
&\hspace{1cm}\text{ and each }k\in\{1,\ldots,n\},
\end{array}
\end{equation}
where $\text{Bid}_{k,\ell}$ and $\text{Ask}_{k,\ell}$ represent the market bid-ask spread of a call option with expiration $T_k$ and strike $\text{Strike}_{k,\ell}$. We emphasize that, in practice, we will have $p_k=0$ for most $k\in\{1,\ldots,n\}$. We take $C := (0,\infty)$ in this computation.

This translates into the framework of \eqref{Equation:MainProblem} by taking $d=1$, $f_k(y_k,y_{k-1}) = (y_k\wedge y_{k-1})\log(y_k/y_{k-1})^2$, and
\[g_k(y_k,y_{k-1}) := \left(\begin{array}{c}
\left(y_k-\text{Strike}_{k,1}\right)^+ - \text{Bid}_{k,1} \\
\text{Ask}_{k,1} - \left(y_k-\text{Strike}_{k,1}\right)^+ \\
\vdots \\
\left(y_k-\text{Strike}_{k,p_k}\right)^+ - \text{Bid}_{k,p_k} \\
\text{Ask}_{k,p_k} - \left(y_k-\text{Strike}_{k,p_k}\right)^+
\end{array}\right)\]
for each $k\in\{1,\ldots,n\}$. Of course, we can also obtain a lower bound by taking $\tilde{f}_n := -f_n$.

In the following, we consider the results of a numerical implementation of the corresponding finite-dimensional approximation given by \eqref{Equation:ApproximateProblem}. Here, we take $x_0 = \$100$, $n=100$, and $T_k = k/240$ for each $k\in\{1,\ldots,100\}$. This is intended to approximate a five-month gamma swap with daily-monitoring\footnote{Here, we simplify to assume twenty equally-spaced business days in each month to avoid dealing with actual day count conventions and trading holiday calendars, although these complications can easily be added for practical applications.} We take $p_k=0$ for all $k$ except $k\in\{40,100\}$, where we have $p_k=7$. We take $\text{Strike}_k = \{\$70,\ldots,\$130\}$ for each $k\in\{40,100\}$ and generate bid-ask spreads from the Black-Scholes pricing formula with $\sigma = 20\%$. Lastly, in these results, we take
\[\Lambda := \{\$1,\$10,\$20,\ldots,\$60,\$70,\$71,\ldots,\$129,\$130,\$140,\ldots,\$190,\$200,\$10000\}.\]
In Table~\ref{Table:GammaSwapSuperHedge}, we illustrate the resulting static hedge positions corresponding to the super-replication strategy for an upper bound. As with the continuously-monitored variance swaps, most of the static hedging positions are placed with 5-month call options. We note the near-the-money hedge positions are approximately those of the variance swap scaled by $100$. The resulting no-arbitrage upper bound on the price of this gamma swap is $\$1.9389$. Similarly, in Table~\ref{Table:GammaSwapSubHedge}, we illustrate the resulting static hedge positions corresponding to the sub-replication strategy for a lower bound. The resulting no-arbitrage lower bound on the price of this gamma swap s $\$1.2443$.

\begin{table}[h]
	\centering
	\begin{tabular}{ |c||c|c|c|c|c|c|c| }
		\hline
		\text{Strike} & \$70 & \$80 & \$90 & \$100 & \$110 & \$120 & \$130 \\ \hline\hline
		$\lambda_1$	& 3.8765 & 0.0001 & 0.0000 & 0.0000 & 0.0000 & 0.0002 & 0.0010 \\ \hline
		$\lambda_2$ & 1.2818 & 0.2468 & 0.2196 & 0.1979 & 0.1800 & 0.1651 & 1.1993 \\ \hline
	\end{tabular}
	\caption{Static positions in call options expiring at times $T_1=1/6$ and $T_2=5/12$, corresponding to a super-hedge of a discretely-monitored gamma swap. The corresponding super-replication value is $\$1.9389$.}
	\label{Table:GammaSwapSuperHedge}
\end{table}

\begin{table}[h]
	\centering
	\begin{tabular}{ |c||c|c|c|c|c|c|c| }
		\hline
		\text{Strike} & \$70 & \$80 & \$90 & \$100 & \$110 & \$120 & \$130 \\ \hline\hline
		$\lambda_1$	& -15.7755 & -0.0005 & 0.0111 & 0.0069 & -0.0206 & -0.0234 & 0.0238 \\ \hline
		$\lambda_2$ & -0.9658 & 0.1504 & 0.1478 & 0.1947 & 0.1555 & 0.0686 & -0.4911 \\ \hline
	\end{tabular}
	\caption{Static positions in call options expiring at times $T_1=1/6$ and $T_2=5/12$, corresponding to a sub-hedge of a discretely-monitored gamma swap. The corresponding sub-replication value is $\$1.2443$.}
	\label{Table:GammaSwapSubHedge}
\end{table}

\section{Proof of Main Results}\label{Section:Proof}

In this section we consider a sequence of results which are used to prove Theorem~\ref{Theorem:MainResult}. There are essentially three main ideas in this section:
\begin{enumerate}
\item The expectation-constrained robust maximization problem can be related to an unconstrained robust maximization problem via standard Lagrangian duality theory, 
\item The solution of a unconstrained robust maximization problem can be written concretely in terms of iterated concave envelopes, and
\item The computation of iterated concave envelopes may be expressed as a single minimization problem.
\end{enumerate}
The second and third idea are contained in the analysis of duality for an unconstrained robust maximization problem, while the first will then be used in the proof of Theorem~\ref{Theorem:MainResult}.

\subsection{Weak Duality for an Unconstrained Robust Maximization Problem}

We start by considering an unconstrained version of the robust maximization problem in \eqref{Equation:MainProblem}. For fixed continuous functions $h_1,\ldots,h_n:\mathbb{R}^{d+1}\to\mathbb{R}$, we define
\begin{equation}\label{Equation:UnconstrainedProblem}
p^\star := \sup\limits_{\mathbb{Q}\in\mathcal{Q}} \mathbb{E}^\mathbb{Q}\left[ \sum\limits_{k=1}^n h_k\left(X_{T_k},\ldots,X_{T_{k-d}}\right)\right].
\end{equation}
As in the setup of the constrained maximization problem \eqref{Equation:MainProblem}, we make the assumption that each $h_1,\ldots,h_n$ may be bounded from above by an affine function.

The goal of this section is to relate $p^\star$ to the following minimization problem:
\begin{equation}\label{Equation:UnconstrainedDualProblem}
\begin{array}{rcl}
d^\star := & \inf\limits_{\phi\in C_0(C^{d+1},\mathbb{R})^n} & \phi_1(x_0,x_0,\ldots,x_{1-d}) \\
& \text{s.t.} & \phi_n \geq h_n \\
&& \phi_k(y_k,\ldots,y_{k-d}) \geq h_k(y_k,\ldots,y_{k-d}) + \phi_{k+1}(y_k,y_k,\ldots,y_{k-d+1})\\
&& \hspace{1cm}\text{for all }k\in\{1,\ldots,n-1\}\text{ and }(y_k,\ldots,y_{k-d})\in C^{d+1}\\
&& \phi_k\text{ is concave in its first entry for each }k\in\{1,\ldots,n\}.
\end{array}
\end{equation}

The intuition here is that \eqref{Equation:UnconstrainedDualProblem} encodes the computation of a sequence of iterated concave envelopes. Our goal is to eventually show that strong duality holds. That is, that $p^\star = d^\star$.

We start by showing that both $p^\star$ and $d^\star$ are finite.

\begin{prop}
Both $p^\star,d^\star<+\infty$.
\end{prop}

\begin{proof}

By assumption, each $h_1,\ldots,h_n$ is bounded from above by an affine function. Then there exists $\alpha\in\mathbb{R}^n$ and $\beta\in\mathbb{R}^{n\times(d+1)}$ such that
\[h_k(y_k,\ldots,y_{k-d}) \leq \alpha_k + \sum\limits_{\ell=0}^d \beta_{k,\ell}y_{k-\ell}\]
for each $k\in\{1,\ldots,n\}$.

Let $\mathbb{Q}\in\mathcal{Q}$ be any martingale measure for $X$. Then we can directly compute
\begin{eqnarray}
\mathbb{E}^\mathbb{Q}\left[\sum\limits_{k=1}^n h_k\left(X_{T_k},\ldots,X_{T_{k-d}}\right)\right] & \leq & \mathbb{E}^\mathbb{Q}\left[\sum\limits_{k=1}^n \left[\alpha_k + \sum\limits_{\ell=0}^d \beta_{k,\ell} X_{T_{k-\ell}}\right]\right] \nonumber\\
& = & \sum\limits_{k=1}^n\left[\alpha_k+\sum\limits_{\ell=0}^d \beta_{k,\ell} x_{(k-\ell)\wedge 0}\right].\nonumber
\end{eqnarray}
This upper-bound is independent of the choice of $\mathbb{Q}$, so we conclude
\[p^\star \leq \sum\limits_{k=1}^n\left[\alpha_k+\sum\limits_{\ell=0}^d \beta_{k,\ell} x_{(k-\ell)\wedge 0}\right] < +\infty.\]

Next, define $\phi_1,\ldots,\phi_n:\mathbb{R}^{d+1}\to\mathbb{R}$ recursively as follows: Let $\phi_n(y_n,\ldots,y_{n-d}) := \alpha_n + \sum_{\ell=0}^d \beta_{n,\ell} y_{n-\ell}$. For each $k\in\{1,\ldots,n-1\}$, let
\[\phi_k(y_k,\ldots,y_{k-d}) := \alpha_k + \sum_{\ell=0}^d \beta_{k,\ell} y_{k-\ell} + \phi_{k+1}(y_k,y_k,\ldots,y_{k-d+1}).\]
By construction, each $\phi_1,\ldots,\phi_n$ is affine (hence concave in the first coordinate) and satisfies the constraints of \eqref{Equation:UnconstrainedDualProblem}. Then $\phi_1,\ldots,\phi_n$ is an admissible choice of functions, so we conclude $d^\star < +\infty$. 

\end{proof}

\begin{prop}
Both $p^\star,d^\star > -\infty$.
\end{prop}

\begin{proof}
The first inequality follows by taking $\mathbb{Q}\in\mathcal{Q}$ to be the trivial martingale measure under which $X_{T_k}=x_0$ for all $k\in\{1,\ldots,n\}$. Then we immediately compute
\begin{eqnarray}
p^\star & \geq & \mathbb{E}^\mathbb{Q}\left[\sum_{k=1}^n h_k\left(X_{T_k},\ldots,X_{T_{k-d}}\right)\right] \nonumber\\
& = & \sum_{k=1}^n h_k\left(x_{k\wedge 0}, x_{(k-1)\wedge 0},\ldots,x_{(k-d)\wedge 0}\right) > -\infty.\nonumber
\end{eqnarray}
For the second inequality, suppose that $\phi_1,\ldots,\phi_n:\mathbb{R}^d\to\mathbb{R}$ is any set of functions satisfying the constraints in \eqref{Equation:UnconstrainedDualProblem}. Then we claim that
\[\phi_\ell\left(x_{\ell\wedge 0},x_{(\ell-1)\wedge 0},\ldots,x_{(\ell-d)\wedge 0}\right) \geq \sum_{k=\ell}^n h_k\left(x_{k\wedge 0},x_{(k-1)\wedge 0},\ldots,x_{(k-d)\wedge 0}\right)\]
for each $\ell\in\{1,\ldots,n\}$.

The case $\ell=n$ follows immediately from the property $\phi_n \geq h_n$. Then suppose that, for some $\ell\in\{1,\ldots,n-1\}$, we know that
\[\phi_{\ell+1}\left(x_{(\ell+1)\wedge 0},x_{\ell\wedge 0},\ldots,x_{(\ell-d+1)\wedge 0}\right) \geq \sum_{k=\ell+1}^n h_k\left(x_{k\wedge 0},x_{(k-1)\wedge 0},\ldots,x_{(k-d)\wedge 0}\right).\]
Recall that $\phi_\ell(y_\ell,\ldots,y_{\ell-d})\geq h_\ell(y_\ell,\ldots,y_{\ell-d})+\phi_{\ell+1}(y_\ell,y_\ell,\ldots,y_{\ell-d+1})$ for all $(y_\ell,\ldots,y_{\ell-d})\in\mathbb{R}^{d+1}$. Then, in particular, we compute
\begin{eqnarray}
\phi_\ell(x_{\ell\wedge 0},x_{(\ell-1)\wedge 0},\ldots,x_{(\ell-d)\wedge 0}) & \geq & h_\ell(x_{\ell\wedge 0},\ldots,x_{(\ell-d)\wedge 0})+\phi_{\ell+1}(x_{\ell\wedge 0},x_{\ell\wedge 0},\ldots,x_{(\ell-d+1)\wedge 0}) \nonumber\\
& = & h_\ell(x_{\ell\wedge 0},\ldots,x_{(\ell-d)\wedge 0})+\phi_{\ell+1}(x_{(\ell+1)\wedge 0},x_{\ell\wedge 0},\ldots,x_{(\ell-d+1)\wedge 0}) \nonumber\\
& \geq & \sum_{k=\ell}^n h_k\left(x_{k\wedge 0},x_{(k-1)\wedge 0},\ldots,x_{(k-d)\wedge 0}\right),\nonumber
\end{eqnarray}
where the equality follows because $\ell\in\{1,\ldots,n-1\}$ implies $\ell\wedge 0=(\ell+1)\wedge 0 =0$. Then by backwards induction on $\ell$, the general claim holds.

In particular, the case $\ell=1$ implies that
\[\phi_1(x_0,x_0,\ldots,x_{1-d}) \geq \sum_{k=1}^n h_k\left(x_{k\wedge 0},x_{(k-1)\wedge 0},\ldots,x_{(k-d)\wedge 0}\right).\]
But because lower bound is independent of choice of $\phi_1,\ldots,\phi_n$, we conclude
\[d^\star \geq \sum_{k=1}^n h_k\left(x_{k\wedge 0},x_{(k-1)\wedge 0},\ldots,x_{(k-d)\wedge 0}\right) > -\infty.\]
\end{proof}

Now, we demonstrate a weak duality relationship between $p^\star$ and $d^\star$.

\begin{lem}\label{Lemma:WeakDuality}
$p^\star \leq d^\star$.
\end{lem}

\begin{proof}

We know $-\infty < d^\star < +\infty$, so for any $\epsilon>0$ there exists functions $\phi_1,\ldots,\phi_n:\mathbb{R}^d\to\mathbb{R}$ which are admissible for \eqref{Equation:UnconstrainedDualProblem} and satisfy
\[d^\star + \epsilon \geq \phi_1(x_0,x_0,\ldots,x_{1-d}).\]

Let $\mathbb{Q}\in\mathcal{Q}$ be an arbitrary martingale measure for $X$. We first claim that
\[\mathbb{E}^\mathbb{Q}\left[\sum_{k=\ell}^n h_k\left(X_{T_k},X_{T_{k-1}},\ldots,X_{T_{k-d}}\right)\right] \leq \mathbb{E}^\mathbb{Q}\left[\phi_\ell\left(X_{T_{\ell-1}},X_{T_{\ell-1}},\ldots,X_{T_{\ell-d}}\right)\right]\]
for each $\ell\in\{1,\ldots,n\}$.

The case $\ell=n$ follows because $\phi_n\geq h_n$ and $\phi_n$ is concave in its first entry. Then we can compute
\begin{eqnarray}
\mathbb{E}^\mathbb{Q}\left[h_n\left(X_{T_n},X_{T_{n-1}},\ldots,X_{T_{n-d}}\right)\right] & \leq & \mathbb{E}^\mathbb{Q}\left[\phi_n\left(X_{T_n},X_{T_{n-1}},\ldots,X_{T_{n-d}}\right)\right] \nonumber\\
& \leq & \mathbb{E}^\mathbb{Q}\left[\phi_n\left(\mathbb{E}^\mathbb{Q}\left[X_{T_n}\mid\mathcal{F}_{T_{n-1}}\right],X_{T_{n-1}},\ldots,X_{T_{n-d}}\right)\right] \nonumber\\
& = & \mathbb{E}^\mathbb{Q}\left[\phi_n\left(X_{T_{n-1}},X_{T_{n-1}},\ldots,X_{T_{n-d}}\right)\right]. \nonumber
\end{eqnarray}
In the second inequality, we applied Jensen’s inequality, and in the following equality, we applied the martingale property of $X$ under $\mathbb{Q}$.

Now, suppose that for some $\ell\in\{1,\ldots,n-1\}$ we know
\[\mathbb{E}^\mathbb{Q}\left[\sum_{k=\ell+1}^n h_k\left(X_{T_k},X_{T_{k-1}},\ldots,X_{T_{k-d}}\right)\right] \leq \mathbb{E}^\mathbb{Q}\left[\phi_{\ell+1}\left(X_{T_{\ell+1}},X_{T_{\ell+1}},\ldots,X_{T_{\ell+1-d}}\right)\right].\]
Recall that $\phi_\ell$ is concave in its first entry and that $\phi_\ell(y_\ell,y_{\ell-1},\ldots,y_{\ell-d})\geq h_\ell(y_\ell,y_{\ell-1},\ldots,y_{\ell-d})+\phi_{\ell+1}(y_\ell,y_\ell,\ldots,y_{\ell-d+1})$ for all $(y_\ell,\ldots,y_{\ell-d})\in\mathbb{R}^{d+1}$. Then applying the same logic as before, we can compute
\begin{eqnarray}
\mathbb{E}^\mathbb{Q}\left[\sum_{k=\ell}^n h_k\left(X_{T_k},X_{T_{k-1}},\ldots,X_{T_{k-d}}\right)\right] & \leq & \mathbb{E}^\mathbb{Q}\left[h_\ell(X_{T_\ell},X_{T_{\ell-1}},\ldots,X_{T_{\ell-d}})\right. \nonumber\\
& & \hspace{2cm} \left. + \phi_{\ell+1}\left(X_{T_{\ell+1}},X_{T_{\ell+1}},\ldots,X_{T_{\ell+1-d}}\right) \right] \nonumber\\
& \leq & \mathbb{E}^\mathbb{Q}\left[\phi_\ell\left(X_{T_\ell},X_{T_{\ell-1}},\ldots,X_{T_{\ell-d}}\right)\right] \nonumber\\
& \leq & \mathbb{E}^\mathbb{Q}\left[\phi_\ell\left(\mathbb{E}^\mathbb{Q}\left[X_{T_\ell}\mid\mathcal{F}_{T_{\ell-1}}\right],X_{T_{\ell-1}},\ldots,X_{T_{\ell-d}}\right)\right] \nonumber\\
& = & \mathbb{E}^\mathbb{Q}\left[\phi_\ell\left(X_{T_{\ell-1}},X_{T_{\ell-1}},\ldots,X_{T_{\ell-d}}\right)\right]. \nonumber
\end{eqnarray}
Then the general statement holds by backwards induction on $\ell$.

Using the $\ell=1$ case, we conclude
\begin{eqnarray}
\mathbb{E}^\mathbb{Q}\left[\sum_{k=1}^n h_k\left(X_{T_k},X_{T_{k-1}},\ldots,X_{T_{k-d}}\right)\right] & \leq & \mathbb{E}^\mathbb{Q}\left[ \phi_1\left(X_{T_0},X_{T_0},\ldots,X_{T_{1-d}}\right) \right] \nonumber\\
& = & \phi_1(x_0,x_0,\ldots,x_{1-d})\nonumber\\
& \leq & d^\star + \epsilon.\nonumber
\end{eqnarray}
However, because $\mathbb{Q}$ was arbitrary and the upper bound is independent of $\mathbb{Q}$, we conclude
\[p^\star \leq d^\star - \epsilon.\]
Because $\epsilon>0$ was arbitrary, the lemma follows.

\end{proof}

\subsection{Strong Duality for an Unconstrained Robust Maximization Problem}

Now we work towards a reverse inequality between $p^\star$ and $d^\star$. We also aim to make clear the relationship between $d^\star$ and a computation of iterated concave envelopes. To this end, we define a sequence of functions $\phi^\star_1,\ldots,\phi^\star_n : \mathbb{R}^{d+1}\to\mathbb{R}$ via the following:
\begin{itemize}
\item The map $y_n\mapsto\phi^\star_n(y_n,y_{n-1},\ldots,y_{n-d})$ is the concave envelope of $y_n\mapsto h_n(y_n,y_{n-1},\ldots,y_{n-d})$,
\item For each $k\in\{1,\ldots,n-1\}$, the map $y_k\mapsto\phi^\star_k(y_k,y_{k-1},\ldots,y_{k-d})$ is the concave envelope of $y_k\mapsto h_k(y_k,y_{k-1},\ldots,y_{k-d}) + \phi^\star_{k+1}(y_k,y_k,\ldots,y_{k-d+1})$.
\end{itemize}
These functions may be be infinite-valued in principle, but we note the assumption on each $h_k$ being bounded above by an affine function is enough to guarantee each is finite-valued.

Before stating the desired lemma, we first recall the following important result about concave envelopes:

\begin{lem}\label{Lemma:ConcaveEnvelopeApproximation}
Fix $\phi:\mathbb{R}\to\mathbb{R}$ and let $\hat\phi:\mathbb{R}\to\mathbb{R}$ denote the concave envelope of $\phi$. For any $y\in\mathbb{R}$ such that $\hat\phi(y)<+\infty$ and any $\epsilon>0$, there exists $p\in[0,1]$ and $z_1,z_2\in\mathbb{R}$ such that
\[y = p\,z_1 + (1-p)\,z_2\]
such that
\[\hat\phi(y) - \epsilon \leq p\,\phi(z_1) + (1-p)\,\phi(z_2).\]
\end{lem}

\begin{proof}
See Corollary 17.1.5 in \cite{Rockafellar1970}.
\end{proof}

This result will be used to construct approximate martingale measures which relate to the functions $\phi^\star_1,\ldots,\phi^\star_n$. With this in hand, we consider the following:

\begin{lem}\label{Lemma:StrongDuality}
$d^\star \leq p^\star$.
\end{lem}

\begin{proof}

We proceed in three steps.

\begin{enumerate}

\item Fix an arbitrary $\epsilon>0$. Then we claim we can construct a martingale measure $\mathbb{Q}\in\mathcal{Q}$ under which
\[\phi^\star_1(x_0,x_0,\ldots,x_{1-d}) \leq \mathbb{E}^\mathbb{Q}\left[\sum_{k=1}^n h_k\left(X_{T_k},\ldots,X_{T_{k-d}}\right)\right] + n\epsilon.\]
Suppose first that we can construct such a measure. Then this would imply
\[d^\star \leq \phi^\star_1(x_0,x_0,\ldots,x_{1-d}) \leq p^\star + n\epsilon,\]
and because $\epsilon>0$ is arbitrary, the lemma would follow.

\item We now construct such a measure $\mathbb{Q}$ via a sequence of transition measures. To start, we specify that the joint measure of $(X_{T_0},\ldots,X_{T_{1-d}})$ is a Dirac measure centered at $(x_0,\ldots,x_{1-d})$.

We then define $\mathbb{Q}$ recursively as follows:

For each $k\in\{1,\ldots,n-1\}$, we define the measure of $X_{T_k}$ conditional on $\mathcal{F}_{T_{k-1}}$ to be the sum of two Dirac measures centered at $z_1$ and $z_2$, with probability $p_1$ and $p_2$ respectively, such that
\[p_1\,z_1+p_2\,z_2 = X_{T_{k-1}}\]
and
\[\phi^\star_k\left(X_{T_{k-1}},X_{T_{k-1}},\ldots,X_{T_{k-d+1}}\right) \leq \epsilon + \sum_{i=1}^2 p_i\,\left[h_k\left(z_i,X_{T_{k-1}},\ldots,X_{T_{k-d+1}}\right)+\phi^\star_{k+1}\left(z_i,z_i,\ldots,X_{T_{k-d+1}}\right)\right].\]
Lastly, we define the measure of $X_{T_n}$ conditional on $\mathcal{F}_{T_{n-1}}$ to be the sum of two Dirac measures centered at $z_1$ and $z_2$ with probability $p_1$ and $p_2$ respectively, such that
\[p_1\,z_1+p_2\,z_2 = X_{T_{n-1}}\]
and
\[\phi^\star_n\left(X_{T_{n-1}},X_{T_{n-1}},\ldots,X_{T_{n-d+1}}\right) \leq \epsilon + \sum_{i=1}^2 p_i\,h_n\left(z_i,X_{T_{n-1}},\ldots,X_{T_{n-d+1}}\right).\]

This makes sense for fixed $(X_{T_{k-1}},\ldots,X_{T+{1-d}})$ using Lemma~\ref{Lemma:ConcaveEnvelopeApproximation}, but, at first glance, could require measurable selection arguments for the resulting measure $\mathbb{Q}$ to be well-defined. We claim this is not the case, but first let us consider why such a measure $\mathbb{Q}$ would satisfy the desired inequality.

This process defines a martingale measure for $X$ by definition as
\[\mathbb{E}^\mathbb{Q}\left[X_{T_k}\mid\mathcal{F}_{T_{k-1}}\right] = p_1\,z_1+p_2\,z_2 = X_{T_{k-1}}.\]
We next claim that
\[\phi^\star_\ell\left(X_{T_{\ell-1}},X_{T_{\ell-1}},\ldots,X_{T_{\ell-d+1}}\right) \leq (n-\ell +1)\epsilon + \mathbb{E}^\mathbb{Q}\left[\sum_{k=\ell}^n h_k\left(X_{T_k},\ldots,X_{T_{k-d}}\right)\mid\mathcal{F}_{T_{\ell-1}}\right],\]
almost-surely, for each $\ell\in\{1,\ldots,n\}$.

The case $\ell=n$ literally follows from the definition of $\mathbb{Q}$. Then suppose that for some $\ell\in\{1,\ldots,n-1\}$ we know that
\[\mathbb{E}^\mathbb{Q}\left[\phi^\star_{\ell+1}\left(X_{T_{\ell}},X_{T_{\ell}},\ldots,X_{T_{\ell-d+2}}\right)\right] \leq (n-\ell)\epsilon + \mathbb{E}^\mathbb{Q}\left[\sum_{k=\ell+1}^n h_k\left(X_{T_k},\ldots,X_{T_{k-d}}\right)\mid\mathcal{F}_{T_\ell}\right],\]
almost-surely. Using this and the definition of $\mathbb{Q}$, we compute
\begin{eqnarray}
\phi^\star_\ell\left(X_{T_{\ell-1}},X_{T_{\ell-1}},\ldots,X_{T_{\ell-d+1}}\right) & \leq & \epsilon + \mathbb{E}^\mathbb{Q}\left[h_\ell\left(X_{T_\ell},\ldots,X_{T_{\ell-d}}\right)\mid\mathcal{F}_{T_{\ell-1}}\right] \nonumber\\
& & \hspace{1cm} +\mathbb{E}^\mathbb{Q}\left[ \phi^\star_{\ell+1}\left(X_{T_{\ell}},X_{T_{\ell}},\ldots,X_{T_{\ell-d+1}}\right)\mid\mathcal{F}_{T_{\ell-1}}\right] \nonumber\\
& \leq & (n-\ell+1)\epsilon + \mathbb{E}^\mathbb{Q}\left[\sum_{k=\ell}^n h_k\left(X_{T_k},\ldots,X_{T_{k-d}}\right)\mid\mathcal{F}_{T_{\ell-1}}\right]\nonumber.
\end{eqnarray}
Then the general claim follows from backwards induction on $\ell$. The case $\ell=1$ then demonstrates the required inequality from the first step.

\item Lastly, we claim that the construction of $\mathbb{Q}$ can be carried out without invoking measurable selection. In particular, we claim that at any stage of the process, the joint measure of $\left(X_{T_k},X_{T_{k-1}},\ldots,X_{T_{1-d}}\right)$ is the sum of finitely-many Dirac measures. Then we are only invoking Lemma~\ref{Lemma:ConcaveEnvelopeApproximation} finitely-many times at each stage and do not require measurable selection.

The case $k=0$ follows directly from the definition of $\mathbb{Q}$. Then, if the joint measure of $\left(X_{T_k},X_{T_{k-1}},\ldots,X_{T_{1-d}}\right)$ is a sum of $m$ Dirac measures, then by definition, the conditional measure of $X_{T_{k+1}}$ given $\mathcal{F}_{T_k}$ is the sum of two Dirac measures. Then the joint measure of $\left(X_{T_{k+1}},X_{T_k},\ldots,X_{T_{1-d}}\right)$ is the sum of at-most $(2m)$ Dirac measures. Then the general result holds immediately by induction on $k$.

\end{enumerate}
\end{proof}

Of course, we then have as an immediate corollary that $p^\star = d^\star$.

\subsection{Proof of Theorem~\ref{Theorem:MainResult}}

Finally, we provide a short proof of Theorem~\ref{Theorem:MainResult} using the duality results of the previous section along with results from standard Lagrange duality theory.

\begin{proof}[Proof of Theorem~\ref{Theorem:MainResult}]

We start by claiming the following inequality holds:
\begin{equation}\label{Equation:LagrangianDuality}
p^\star \leq \inf\limits_{\lambda\geq 0}D(\lambda) := \inf\limits_{\lambda \geq 0} \sup\limits_{\mathbb{Q}\in\mathcal{Q}}\mathbb{E}^\mathbb{Q}\left[ \sum\limits_{k=1}^n f_k\left(X_{T_k},\ldots,X_{T_{k-d}}\right)+\lambda_k\cdot g_k\left(X_{T_k},\ldots,X_{T_{k-d}}\right) \right].
\end{equation}
This is a result of applying standard Lagrangian duality theory to \eqref{Equation:MainProblem} (see \cite{Rockafellar1970,Boyd2004}). Furthermore, we claim that if there exists $\mathbb{Q}\in\mathcal{Q}$ such that $\mathbb{E}^\mathbb{Q}\left[g_k\left(X_{T_k},\ldots,X_{T_{k-d}}\right)\right]\geq 0$ component-wise for each $k\in\{1,\ldots,n\}$, then we have equality above. This is Slater's condition in the special case of affine constraints, which can be found in the previous references or in \cite{Zalinescu2002,Ekeland1983} for specifics of the infinite-dimensional case.

For any fixed $\lambda\in\mathbb{R}^{p_1}\times\cdots\times\mathbb{R}^{p_n}$, define $h^\lambda_1,\ldots,h^\lambda_n : \mathbb{R}^{d+1}\to\mathbb{R}$ as
\[h^\lambda_k(y_k,\ldots,y_{k-d}) := f_k(y_k,\ldots,y_{k-d})+\lambda_k\cdot g_k(y_k,\ldots,y_{k-d})\]
for each $k\in\{1,\ldots,n\}$. If $f_1,\ldots,f_n$ and $g_1,\ldots,g_n$ are bounded from above by affine functions, then so are $h_1,\ldots,h_n$. Then by Lemma~\ref{Lemma:WeakDuality} and Lemma~\ref{Lemma:StrongDuality}, we have
\begin{equation}\label{Equation:LagrangianDuality2}
\begin{array}{rccl}
D(\lambda)& = & \sup\limits_{\mathbb{Q}} & \mathbb{E}^\mathbb{Q}\left[\sum\limits_{k=1}^n h^\lambda_k\left(X_{T_k},\ldots,X_{T_{k-d}}\right)\right] \\
 & = & \inf\limits_{\phi\in C_0(C^{d+1},\mathbb{R})^n} & \phi_1(x_0,x_0,\ldots,x_{1-d}) \\
&& \text{s.t.} & \phi_n \geq h^\lambda_n \\
&&& \phi_k(y_k,\ldots,y_{k-d}) \geq h^\lambda_k(y_k,\ldots,y_{k-d}) + \phi_{k+1}(y_k,y_k,\ldots,y_{k-d+1})\\
&&& \hspace{1cm}\text{for all }k\in\{1,\ldots,n-1\}\text{ and }(y_k,\ldots,y_{k-d})\in C^{d+1}\\
&&& \phi_k\text{ is concave in its first entry for each }k\in\{1,\ldots,n\}.
\end{array}
\end{equation}
Putting together \eqref{Equation:LagrangianDuality} and \eqref{Equation:LagrangianDuality2}, and introducing $h^\lambda_1,\ldots,h^\lambda_n$ as auxiliary variables in the minimization for convenience of notation, we obtain the stated result.

\end{proof}

\section{Discussion and Conclusion}

The key contribution of this paper is showing how to re-write a expectation-constrained robust maximization problem in a dual minimization form which intuitively encodes the computation of iterated concave envelopes. One immediate application is to obtain model-free no-arbitrage bounds on the prices of several exotic financial derivatives, but we note that the ideas may have broader application.

There is a broad literature on constrained stochastic control. One common approach is to introduce extra state variables corresponding to the constraint and directly solve a state-constrained optimal control problem \cite{Bouchard2012,Rokhlin2014,Katsoulakis1994}, although rigorously proving dynamic programming results can pose technical difficulties. There are particular technical difficulties introduced by our consideration of a maximization over martingale measures, as illustrated in \cite{Bayraktar2014,Ekren2014}. An alternative method is to re-write the original problem as a related multi-level optimization problem, as in \cite{Miller2015,Ankirchner2015,Bayraktar2016}, but this depends heavily on the structure of the problem. Lastly, many authors take a Lagrangian approach and convert the constrained maximization problem to a convex minimax problem \cite{Horiguchi2001,Makasu2009,Lopez1995}, although there are certainly theoretical issues with computing sub-gradients as illustrated in \cite{Miller2015b}. The contents of this paper are an addition to the latter approach, and suggest how to convert to a convex minimization problem when additional strong duality results holds.

The approach of this paper may initially seem different than much of the literature on super-replication, which often views the problem as an optimal control problem \cite{Galichon2014,Cox2015,Bonnans2013,Soner2003,Bentahar2006} or optimal martingale transport problem \cite{Beiglbock2013,Dolinsky2014,DeMarco2015}. Instead, we take the approach of viewing super-replication as a linear program subject to constraints, in the spirit of \cite{Acciaio2013,Riedel2011,Bayraktar2016b}. We note, however, there is an intuitive connection between a viscosity approach and ours, in that the minimization problem is analogous to the minimization over all viscosity super-solutions in the Perron method \cite{Ishii1987,Crandall1992}. We expect that in future research, this more general concept may be used to create minimization schemes for computing no-arbitrage price bounds for various types of exotic derivatives with path-dependent pay-offs requiring the introduction of additional state-variables.

%\section*{Acknowledgements}

{\small 

\bibliographystyle{siam}
\bibliography{Bibliography}

\def\cprime{$'$}
\begin{thebibliography}{10}

\bibitem{Acciaio2013}
{\sc B.~Acciaio, M.~Beiglb{\"o}ck, F.~Penkner, and W.~Schachermayer}, {\em A
  model-free version of the fundamental theorem of asset pricing and the
  super-replication theorem}, Mathematical Finance,  (2013).

\bibitem{Andersen2000}
{\sc E.~D. Andersen and K.~D. Andersen}, {\em The mosek interior point
  optimizer for linear programming: an implementation of the homogeneous
  algorithm}, in High performance optimization, Springer, 2000, pp.~197--232.

\bibitem{Ankirchner2015}
{\sc S.~Ankirchner, M.~Klein, and T.~Kruse}, {\em A verification theorem for
  optimal stopping problems with expectation constraints},  (2015).

\bibitem{Bayraktar2016}
{\sc E.~Bayraktar and C.~W. Miller}, {\em Distribution-constrained optimal
  stopping}, arXiv preprint arXiv:1604.03042,  (2016).

\bibitem{Bayraktar2014}
{\sc E.~Bayraktar and S.~Yao}, {\em On the robust optimal stopping problem},
  SIAM J. Control Optim., 52 (2014), pp.~3135--3175.

\bibitem{Bayraktar2016b}
{\sc E.~Bayraktar and Z.~Zhou}, {\em Super-hedging american options with
  semi-static trading strategies under model uncertainty}, Available at SSRN
  2785625,  (2016).

\bibitem{Beiglbock2013}
{\sc M.~Beiglb{\"o}ck, P.~Henry-Labord{\`e}re, and F.~Penkner}, {\em
  Model-independent bounds for option prices---a mass transport approach},
  Finance Stoch., 17 (2013), pp.~477--501.

\bibitem{Bentahar2006}
{\sc I.~Bentahar and B.~Bouchard}, {\em Barrier option hedging under
  constraints: a viscosity approach}, SIAM Journal on Control and Optimization,
  45 (2006), pp.~1846--1874.

\bibitem{Bonnans2013}
{\sc J.~F. Bonnans and X.~Tan}, {\em A model-free no-arbitrage price bound for
  variance options}, Appl. Math. Optim., 68 (2013), pp.~43--73.

\bibitem{Bouchard2012}
{\sc B.~Bouchard and M.~Nutz}, {\em Weak dynamic programming for generalized
  state constraints}, SIAM J. Control Optim., 50 (2012), pp.~3344--3373.

\bibitem{Boyd2004}
{\sc S.~Boyd and L.~Vandenberghe}, {\em Convex optimization}, Cambridge
  University Press, Cambridge, 2004.

\bibitem{Carr2010}
{\sc P.~Carr and R.~Lee}, {\em Hedging variance options on continuous
  semimartingales}, Finance and Stochastics, 14 (2010), pp.~179--207.

\bibitem{Cox2015}
{\sc A.~M. Cox and S.~K{\"a}llblad}, {\em Model-independent bounds for asian
  options: a dynamic programming approach}, arXiv preprint arXiv:1507.02651,
  (2015).

\bibitem{Crandall1992}
{\sc M.~G. Crandall, H.~Ishii, and P.-L. Lions}, {\em User’s guide to
  viscosity solutions of second order partial differential equations}, Bulletin
  of the American Mathematical Society, 27 (1992), pp.~1--67.

\bibitem{DeMarco2015}
{\sc S.~De~Marco and P.~Henry-Labordere}, {\em Linking vanillas and vix
  options: A constrained martingale optimal transport problem}, SIAM Journal on
  Financial Mathematics, 6 (2015), pp.~1171--1194.

\bibitem{Deng2016}
{\sc S.~Deng and X.~Tan}, {\em Duality in nondominated discrete-time models for
  americain options}, arXiv preprint arXiv:1604.05517,  (2016).

\bibitem{Dolinsky2014}
{\sc Y.~Dolinsky and H.~M. Soner}, {\em Martingale optimal transport and robust
  hedging in continuous time}, Probability Theory and Related Fields, 160
  (2014), pp.~391--427.

\bibitem{Ekeland1983}
{\sc I.~Ekeland and T.~Turnbull}, {\em Infinite-dimensional optimization and
  convexity}, University of Chicago Press, 1983.

\bibitem{Ekren2014}
{\sc I.~Ekren, N.~Touzi, and J.~Zhang}, {\em Optimal stopping under nonlinear
  expectation}, Stochastic Process. Appl., 124 (2014), pp.~3277--3311.

\bibitem{Galichon2014}
{\sc A.~Galichon, P.~Henry-Labord{\`e}re, and N.~Touzi}, {\em A stochastic
  control approach to no-arbitrage bounds given marginals, with an application
  to lookback options}, Ann. Appl. Probab., 24 (2014), pp.~312--336.

\bibitem{Hobson2015}
{\sc D.~Hobson and M.~Klimmek}, {\em Robust price bounds for the forward
  starting straddle}, Finance and Stochastics, 19 (2015), pp.~189--214.

\bibitem{Horiguchi2001}
{\sc M.~Horiguchi}, {\em Markov decision processes with a stopping time
  constraint}, Mathematical methods of operations research, 53 (2001),
  pp.~279--295.

\bibitem{Ishii1987}
{\sc H.~Ishii et~al.}, {\em Perron’s method for hamilton-jacobi equations},
  Duke math. J, 55 (1987), pp.~369--384.

\bibitem{Kahale2012}
{\sc N.~Kahal{\'e}}, {\em Super-replication of financial derivatives via convex
  programming}, Available at SSRN 2172315,  (2012).

\bibitem{Karmarkar1984}
{\sc N.~Karmarkar}, {\em A new polynomial-time algorithm for linear
  programming}, in Proceedings of the sixteenth annual ACM symposium on Theory
  of computing, ACM, 1984, pp.~302--311.

\bibitem{Katsoulakis1994}
{\sc M.~A. Katsoulakis}, {\em Viscosity solutions of second order fully
  nonlinear elliptic equations with state constraints}, Indiana Univ. Math. J.,
  43 (1994), pp.~493--519.

\bibitem{Khachiyan1980}
{\sc L.~G. Khachiyan}, {\em Polynomial algorithms in linear programming}, USSR
  Computational Mathematics and Mathematical Physics, 20 (1980), pp.~53--72.

\bibitem{Lee2010}
{\sc R.~Lee}, {\em Gamma swap}, Encyclopedia of Quantitative Finance,  (2010).

\bibitem{Lopez1995}
{\sc F.~L{\'o}pez, M.~San~Miguel, and G.~Sanz}, {\em Lagrangean methods and
  optimal stopping}, Optimization, 34 (1995), pp.~317--327.

\bibitem{Makasu2009}
{\sc C.~Makasu}, {\em Bounds for a constrained optimal stopping problem},
  Optimization Letters, 3 (2009), pp.~499--505.

\bibitem{Miller2015}
{\sc C.~W. Miller}, {\em Non-linear pde approach to time-inconsistent optimal
  stopping}, arXiv preprint arXiv:1510.05766,  (2015).

\bibitem{Miller2015b}
{\sc C.~W. Miller and I.~Yang}, {\em Optimal control of conditional
  value-at-risk in continuous time}, arXiv preprint arXiv:1512.05015,  (2015).

\bibitem{Nesterov1994}
{\sc Y.~Nesterov and A.~Nemirovskii}, {\em Interior-point polynomial algorithms
  in convex programming}, vol.~13, Siam, 1994.

\bibitem{Neufeld2013}
{\sc A.~Neufeld, M.~Nutz, et~al.}, {\em Superreplication under volatility
  uncertainty for measurable claims}, Electron. J. Probab, 18 (2013),
  pp.~1--14.

\bibitem{Nutz2014}
{\sc M.~Nutz}, {\em Superreplication under model uncertainty in discrete time},
  Finance and Stochastics, 18 (2014), pp.~791--803.

\bibitem{Riedel2011}
{\sc F.~Riedel}, {\em Finance without probabilistic prior assumptions}, arXiv
  preprint arXiv:1107.1078,  (2011).

\bibitem{Rockafellar1970}
{\sc R.~T. Rockafellar}, {\em Convex analysis}, Princeton Mathematical Series,
  No. 28, Princeton University Press, Princeton, N.J., 1970.

\bibitem{Rokhlin2014}
{\sc D.~B. Rokhlin}, {\em Stochastic {P}erron's method for optimal control
  problems with state constraints}, Electron. Commun. Probab., 19 (2014),
  pp.~no. 73, 15.

\bibitem{Soner2003}
{\sc H.~M. Soner and N.~Touzi}, {\em The problem of super-replication under
  constraints}, in Paris-{P}rinceton {L}ectures on {M}athematical {F}inance,
  2002, vol.~1814 of Lecture Notes in Math., Springer, Berlin, 2003,
  pp.~133--172.

\bibitem{Yen2015}
{\sc I.~E.-H. Yen, K.~Zhong, C.-J. Hsieh, P.~K. Ravikumar, and I.~S. Dhillon},
  {\em Sparse linear programming via primal and dual augmented coordinate
  descent}, in Advances in Neural Information Processing Systems, 2015,
  pp.~2368--2376.

\bibitem{Zalinescu2002}
{\sc C.~Zalinescu}, {\em Convex analysis in general vector spaces}, World
  Scientific, 2002.

\end{thebibliography}
}
\end{document}